\documentclass[twoside,leqno,twocolumn]{article}  
\usepackage{ltexpprt,epsfig} 

\newcommand{\mClass}[0]{{\cal X}}
\newcommand{\dClass}[0]{{\cal D}}
\newcommand{\fPlace}[0]{{f_{P}}}
\newcommand{\fRecon}[0]{{f_{R}}}
\newcommand{\deleted}[1]{}

\title{\Large 
Geometric reconstruction from point-normal data
}
\author{Eleanor G. Rieffel \\
FXPAL \\
rieffel@fxpal.com
\and 
Don Kimber \\
FXPAL \\
kimber@fxpal.com
\and
Jim Vaughan \\
FXPAL \\
jimv@fxpal.com
}

\date{}

\begin{document}
\maketitle


\begin{abstract} \small\baselineskip=9pt 
Creating virtual models of real spaces and objects
is cumbersome and time consuming. 
This paper focuses on the problem of geometric reconstruction 
from sparse data obtained from certain image-based modeling approaches.
A number of elegant and simple-to-state problems arise concerning 
when the geometry can be reconstructed. 
We describe results and counterexamples, and list open problems.
\end{abstract}

\section{Introduction.}

While three-dimensional virtual models have long been used in
industry for design, the increased speed and graphics
capabilities of today's computers, higher bandwidth, and the
popularity of virtual environments mean that virtual models
are becoming ever easier to view, manipulate, and distribuite. 
This improved ease of use has spawned an increasing desire for
better methods to create models, including models of real objects and spaces.

At FXPAL, we are particularly interested in the use of virtual models
in a factory setting \cite{Back09} and in surveillance 
\cite{Girgensohn07, Rieffel07}. 
Common applications include training, immersive telepresence, 
military exercises, and design and testing of emergency response plans. 
Other applications range from virtual tourism \cite{Snavely06, Shao07} 
and psychiatric treatment for post-traumatic stress disorder \cite{Difede02}, 
phobias \cite{Moore02}, and autism \cite{Strickland97}.
Real estate offices are beginning to use
three-dimensional models to support the creation of virtual 
tours \cite{Heartwood}.
Not only are marketing departments beginning to make models of their
products available to potential purchasers, but applications are
springing up around these models. For example, MyDeco \cite{MyDeco} 
enables users to create models of a three-dimensional space, 
place models of real furniture and other home accessories 
that are available for purchase in the virtual space, and then buy  
any of these products directly from the site. Virtual worlds such as 
Second Life \cite{SecondLife} are filled with more or less realistic
models of real places and objects.
Google Earth \cite{GoogleEarth} now includes three-dimensional
models of various buildings. 

Unfortunately, creating virtual models of real objects and spaces
remains cumbersome and time consuming. 
Current state of the art modeling is done by artists using
interactive modeling tools, often supported by measurement
and photographs of the real space.
An ambitious long term research goal is to
automatically construct such models from collected images;
fully automatic approaches are not yet possible.
FXPAL's Pantheia system \cite{Kimber09, Rieffel09}
enables users to create models by {\it marking up} 
the real world with pre-printed uniquely identifiable markers. 
Predefined meanings associated with the markers guide the system 
in creating models. The position and outward pointing normal at
each marker can be estimated from user-captured images or video of 
the marked-up space. Point-normal data, consisting of the position and
outward pointing normal, can be obtained using other technologies
such as range scanners.

This paper focuses on the problem of reconstructing the geometry 
from the marker information. Our initial attempts at reconstruction
used {\it ad hoc} reconstruction algorithms and markup placement strategies. 
When we tried to model a new space, we often needed to place 
additional markers, add meanings to the markup language, or revise 
the reconstruction algorithm to make it more powerful. This paper is 
the result of our work to place the geometric reconstruction aspect 
of our system on a firmer formal footing.  
\begin{figure}
\centering
\includegraphics[width=1.1in,height=1.8in,viewport=130 55 430 600,clip]{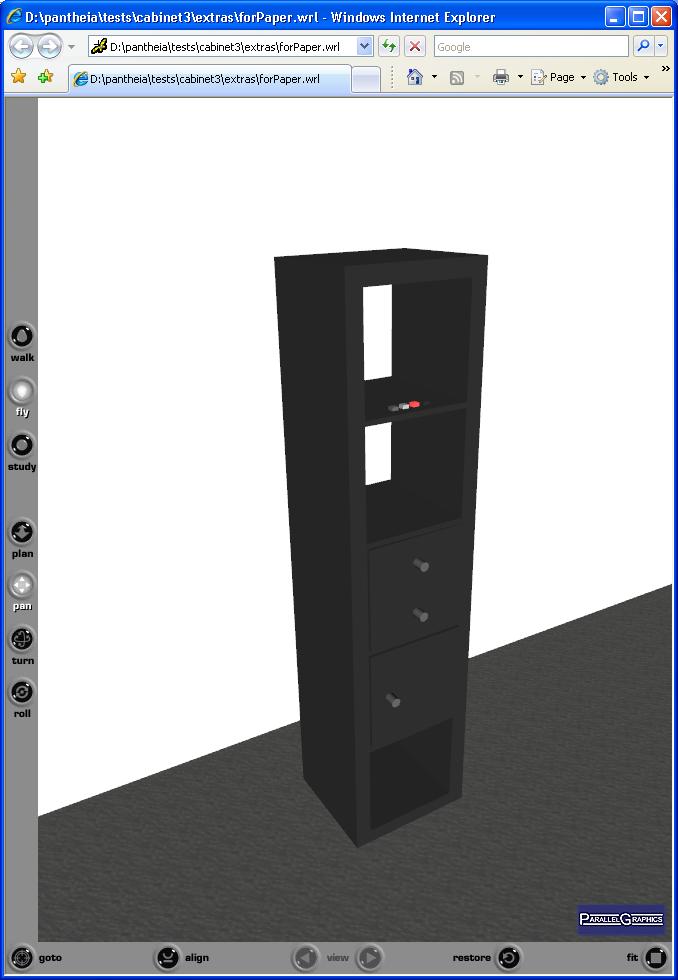}
\includegraphics[width=1.1in,height=1.8in]{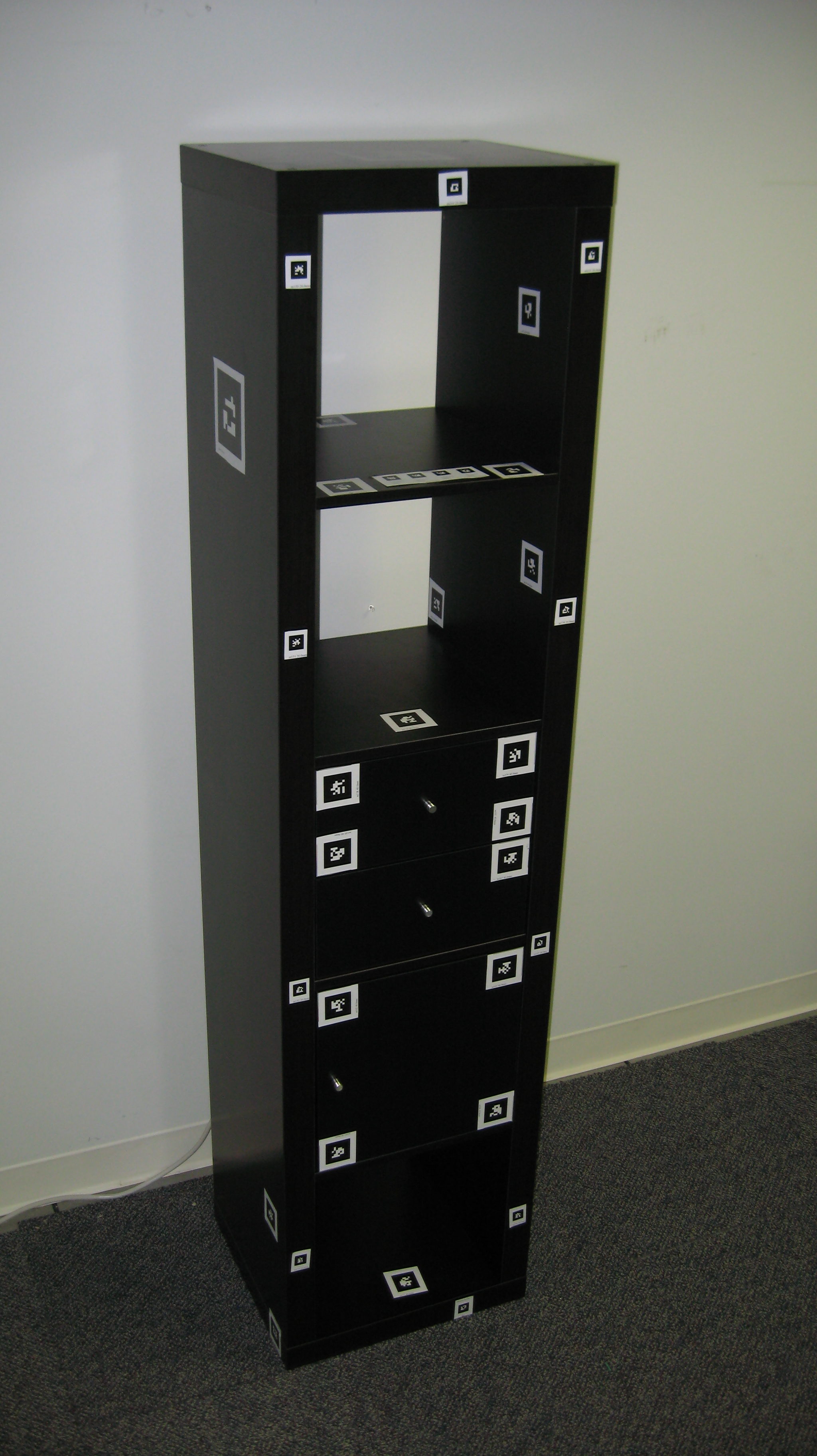}
\caption{\label{figure::cabinet} Model of an IKEA Bookcase cabinet generated
by the Pantheia system. }
\end{figure}


\section{Related Work.}
\label{sec:relWork}

This section discusses two types of related work. First, we 
discuss related work in the area of image-based modeling.
Then we survey previous work in
polyhedral reconstruction from simple geometric data.

Researchers such as
\cite{PollefeysTutorial, Pollefeys02, Dellaert00}
work on non-marker-based methods for constructing models from images.
Their work advances progress on the hard problem of deducing
geometric structure from image features.
Instead, we make the problem simpler by placing markers
that are easily detected and have meanings that greatly simplify
the geometric deduction. Furthermore, a marker-based approach
enables users to specify which parts of the scene are important.
In this way, Pantheia
handles clutter removal and certain occlusion issues easily, since
it renders what the markers indicate rather than what is seen.

From a large number of photographs of a place or object,
visual features, such as SIFT features \cite{Lowe99}, 
can be extracted and rendered as point cloud models
\cite{Snavely06, Snavely08}. More generally, 
the area of `point based graphics' provides methods for
representing surfaces by point data, without requiring other 
graphics primitives such as meshes \cite{GrossPfister07}.  These methods 
have been used as primitives for modeling tools
\cite{pointset3D}. Amenta {\it et al.}~\cite{Amenta} describe the point 
based notion of `surfels' which are points and normals. Our markers 
specify one dimension over a surfel: the orientation of the marker 
within its plane. Generally, point based methods use large numbers 
of surface points, and aim to produce smooth surfaces.
Our system aims to produce polyhedral models with low polygon count
from sparse geometric data.

There is a rich history of work on reconstructing polyhedra
from partial descriptions. See Lucier \cite{Lucier} for a survey. 
There is little work, however, on reconstructing polyhedra from 
sparse point-plane or  point-normal data, let alone more complex 
metadata. Biedl {\it et al.} discuss several polygon reconstruction 
problems based on point-normal data and related data \cite{Biedl2008}. 
Their reconstruction results are limited to two dimensions.

\section{Overview of the Pantheia System.}

\begin{figure}
\centering
\includegraphics[width=3.5in,height=2.4in,viewport=0 60 750 550,clip]{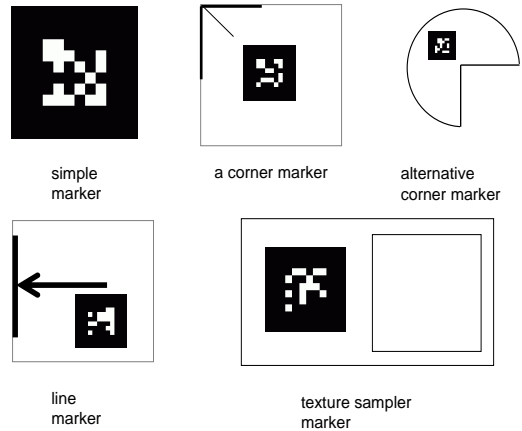}
\caption{\label{figure:markerExamples}Examples of visual markup markers}
\end{figure}

For markers, the Pantheia system uses the two-dimensional 
two dimensional barcode style markers of ARToolKitPlus \cite{ARToolKit}
(see Figure \ref{figure:markerExamples}).
Pantheia \cite{Rieffel09,Kimber09} takes as input a set of images, identifies
the markers in each image, and determines the 
relative pose of the marker to camera for each marker in each image. 
If the pose of a marker in the world is known then the 
pose of the camera for each image
can be determined. Conversely, if the pose of a camera is known  then
the pose of any marker identified in that image can be calculated.
From a set of images that meets a few simple conditions, 
the pose of every marker and the pose of the camera for every image 
can be estimated. Pantheia
obtains good estimates using ARToolKitPlus \cite{ARToolKit}
together with the sparse bundle adjustment package SBA \cite{SBA} that
globally optimizes the pose estimates for all markers and images.

Pantheia creates models using a markup language \cite{Kimber09}
that includes elements specifying 
appearance characteristics, interactive elements, and geometric 
properties of the scene.
This paper concerns only the geometric reconstruction aspect
of the system supported by the following marker language elements: 
\begin{itemize}
\item planar (plane, wall, ceiling, floor, door), 
\item shape (parametric shape, extrusion), 
\item modifier (edge, corner). 
\end{itemize}

This paper reports on results of our ongoing effort 
to define a simple yet powerful markup language and to develop robust 
reconstruction algorithms that together support the creation of a large
class of virtual models. The output of our system can be thought of as the 
description of a virtual model and its dynamic capabilities in terms 
of a language such as COLLADA \cite{collada}, which supports the 
expression of physics, or of a language such as VRML \cite{vrml} together 
with physics specifications for a physics engine such as ODE \cite{ODE}
in order to support interactive elements. Pantheia 
currently saves models as VRML together with metadata files that 
specify relations between named parts of the VRML scene graph. 

\label{robustness}
The user is encouraged, when placing markers on the same plane, to
use markers that have the same plane ID. Furthermore, planes that 
are close to aligned are forced to align by averaging their normal vectors. 
Finally, there is an option 
that takes all planes close to being aligned with one of the coordinate
planes and snaps them to being precisely parallel to the coordinate plane. 
These features all help with the robustness of the estimation.
Accuracy results for an early version of the system are reported 
in \cite{Kimber09}. More details on the Pantheia system can be
found in both \cite{Kimber09,Rieffel09}.

As Pantheia's designers, we can choose
which markup strategies to suggest to users, which marker meanings
to make available, and which reconstruction algorithm to use. 
Our aim is to design markup strategies that, for a broad class of models,
specify a unique model, are easy for users to understand and carry out,
and have an efficient reconstruction algorithm.

\begin{figure}
\centering
\includegraphics[width=3.1in,height=1.8in]{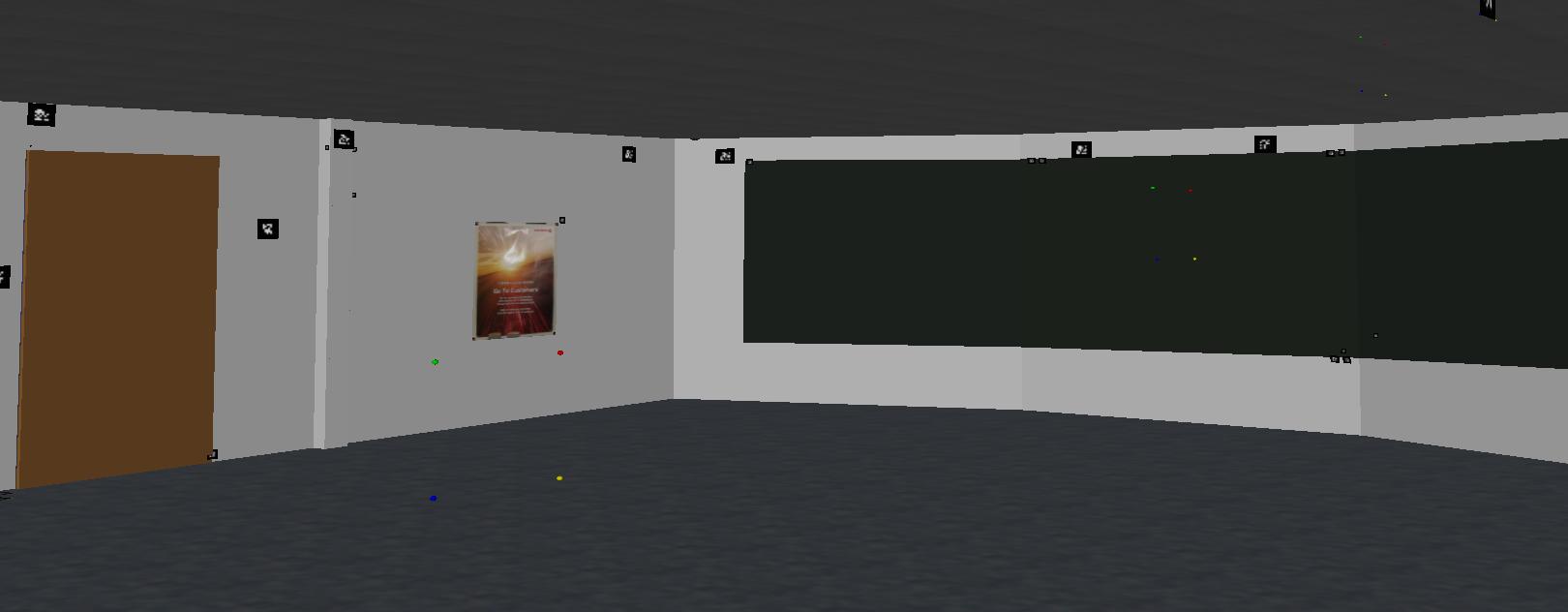}
\includegraphics[width=3.1in,height=1.8in]{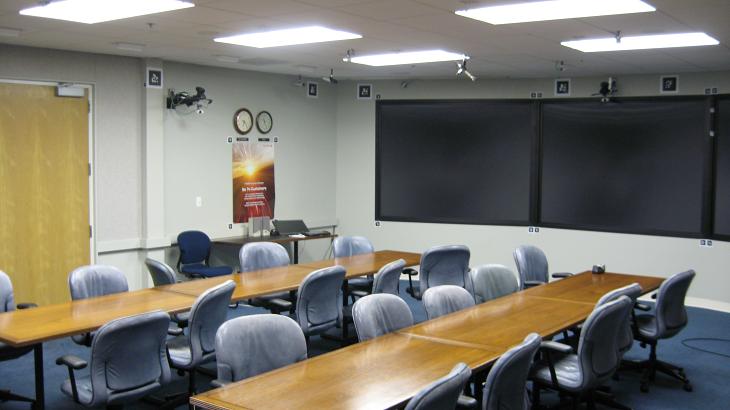}
\caption{\label{figure::kumoModel}Pantheia generated model of a room}
\end{figure}


\section{Formal Framework.}
\label{sec:formalFramework}

This section contains definitions used through the paper, and sets up
the formal framework in which we discuss markup strategies and
reconstruction algorithms. The framework is somewhat abstract in that
we assume that the position and orientation of all markers is known. 
When marker poses are obtained using vision techniques, some
markers may not be visible to a camera. For those markers, a pose 
cannot be obtained. We do not handle these cases, nor do we discuss
robustness questions. Both are interesting areas for future work.

Section \ref{basicDefs} begins with some
basic geometric definitions. Section \ref{markupDefs} formally
defines a {\em markup description} and related concepts.
Section \ref{markerData} defines various types of marker data 
considered in this paper. Section \ref{markupStrategies} discusses
the markup strategies considered in this paper.
Section \ref{elaboration} covers elaboration techniques for
constructing more complex polyhedra from simpler ones.

\subsection{Basic geometric definitions.}
\label{basicDefs}

We begin with some basic definitions. 
This paper is primarily concerned with
constructing polyhedra from marker data. Surprisingly, there are a number
of variations as to how the terms ``polygon" and  ``polyhedron" are defined. 
For clarity, we state the definitions we use in this paper.

\begin{Definition}
{\rm A {\em polygon} is a closed, connected, two-dimensional region of a plane
whose boundary consists of a finite set of segments, $e_1, e_2, ... e_n$,
called {\em edges} with endpoints, $v_1, v_2, ... v_n$,
called {\em vertices}.  Furthermore, every vertex $v_i$  is the endpoint of two
edges, each edge ends in two vertices, and two edges can intersect only in 
a vertex.}
\end{Definition}

\begin{Definition}
{\rm A {\em polyhedron, pl.~polyhedra,} is a closed, connected, 
three-dimensional volume of three-dimensional Euclidean space 
whose boundary consists
of a finite set of polygons called {\em faces} such that every
edge of every polygon is shared with exactly one other polygon, 
two faces may intersect only at edges or vertices shared by the
two polygons, and the faces that share a vertex can be ordered 
so that face $f_i$ shares an edge with $f_{i+1}$ (mod $q$, 
the number of such faces).}
\end{Definition}

This definition allows polyhedra to have holes. Our theorems hold
for polyhedra of arbitrary genus.

\begin{Definition} {\rm A region is {\em convex} if, for every pair of
points in the region, the line segment connecting those points
is entirely contained in the region.}
\end{Definition}

\begin{Definition}
A polyhedron is {\em orthogonal} if all of its faces 
are parallel to one of the three coordinate planes.
{\rm A set of polyhedra is {\em orthogonal} if all
of the polyhedra in the set are orthogonal.}
\end{Definition}

\subsection{Markup definitions.}
\label{markupDefs}

We model a marker as a numeric identifier together with metadata.  
Minimally the metadata includes the position of each marker. 
The metadata may include other information about the marker or its 
placement, such as the outward pointing normal, or meanings from the 
markup language associated with sets of markers. We assume that all 
marker information is known -- that all markers have been seen and
robustly estimated.

\begin{Definition}
{\rm A {\em markup description} of a model is a set $I_M$ of marker 
indices, together with the metadata associated with those markers.}
\end{Definition}

Let $\mClass$ be the class of models under consideration.

\begin{Definition}
{\rm A {\em markup placement strategy} $\fPlace$ is a mapping
$\fPlace : \mClass \rightarrow \dClass$ from the class of models
$\mClass$ to the space of markup descriptions $\dClass$.}
\end{Definition}

A markup strategy may or may not be deterministic.  For
example, a deterministic strategy might be, ``for each face of the
model, place a marker on the face at the centroid of its vertices".
A nondeterministic strategy might be, ``For each face of the model,
place a marker somewhere on the face".

\begin{Definition}
{\rm A {\em markup reconstruction rule} $\fRecon$ is a mapping from the
class of markup descriptions $\dClass$ to the class of models $\mClass$. 
A {\em markup reconstruction algorithm} is a procedure which
implements a reconstruction rule.}
\end{Definition}

\begin{Definition}
{\rm A {\em markup system} is a markup placement strategy $\fPlace$ 
together with a reconstruction rule $\fRecon$.  A markup system is 
complete and faithful if for
every $m \in \mClass$, $\fRecon( \fPlace( m )) = m$.}
\end{Definition}

\subsection{Marker data and metadata types.}
\label{markerData}

This paper considers three main types of marker data.

\begin{Definition}
{\rm {\em Point marker data} consists of the position of the center
of the marker, or the position of a point specified by the metadata
relative to the center of the marker. The corner marker shown in
Figure \ref{figure:markerExamples} is an example of a marker that specifies
a position that is not the center of the marker.}
\end{Definition}

\begin{Definition}
{\rm {\em Point-plane metadata} consists of the position
of the marker and the plane in which it lies.}
\end{Definition}

\begin{Definition}
{\rm {\em Point-normal metadata} consists of the position and outward
pointing normal at each marker. }
\end{Definition}

To each of these basic types, various levels of additional  metadata can 
be added. Useful types of additional metadata include IDs that indicate
that all markers with that ID share a property such as being on the
same face or defining the same polyhedron, orderings of a set of markers
that indicate, for example, the order in which to traverse the vertices of
a face, and relationships, such as two faces share an edge.

\subsection{Some Markup Strategies.}
\label{markupStrategies}

\begin{Definition}
{\rm The {\em marker-per-face} markup strategy is any placement of at
least one marker on each face.} 
\end{Definition}

Most of this paper will discuss reconstruction from a marker-per-face
strategy with point-plane or point-normal metadata and possibly 
additional metadata. As mentioned in Section \ref{sec:relWork},
much more common in the literature are discussions of reconstruction
from a vertex markup strategy in which every vertex is marked. We are 
less interested in such strategies than marker-per-face strategies
because they require more precision in placement on the part of 
a user and require markers to be placed in places that may be
out of reach or even hidden. In Section \ref{sec:other}, we discuss
a few results related to vertex or edge markup strategies. 

\subsection{Elaboration.}
\label{elaboration}

Elaboration is a way to create a more complex polyhedron from a 
base polyhedron by gluing a polyhedron to a face or removing a
polyhedron aligned with the face from the base polyhedron.
{\it Extrusions} and {\it intrusions} are special cases of 
elaboration. 
The reconstruction results related to elaboration focus on
reconstruction of polyhedra that can be obtained by
elaborating a convex polyhedra with extrusions and intrusions 
of convex polygons.


\begin{Definition}
{\rm {\em Orthogonal Polygonal Extrusion} or just {\em Extrusion: } 
A polygon in the interior of a face of the base polyhedron
is extruded outward, perpendicular to the face, a constant amount.}
\end{Definition}

\begin{Definition}
{\rm {\em Orthogonal Polygonal Intrusion} or just {\em Intrusion: } 
A polygon in the interior of a face is pushed inward, 
perpendicular to the face, a constant amount.}
\end{Definition}

We will also consider polyhedra obtained by taking the union
of separately defined but intersecting polyhedra, including
ones obtained by {\em gluing} along a shared face.

\begin{figure}
\centering
\includegraphics[width=3.0in,height=2.0in]{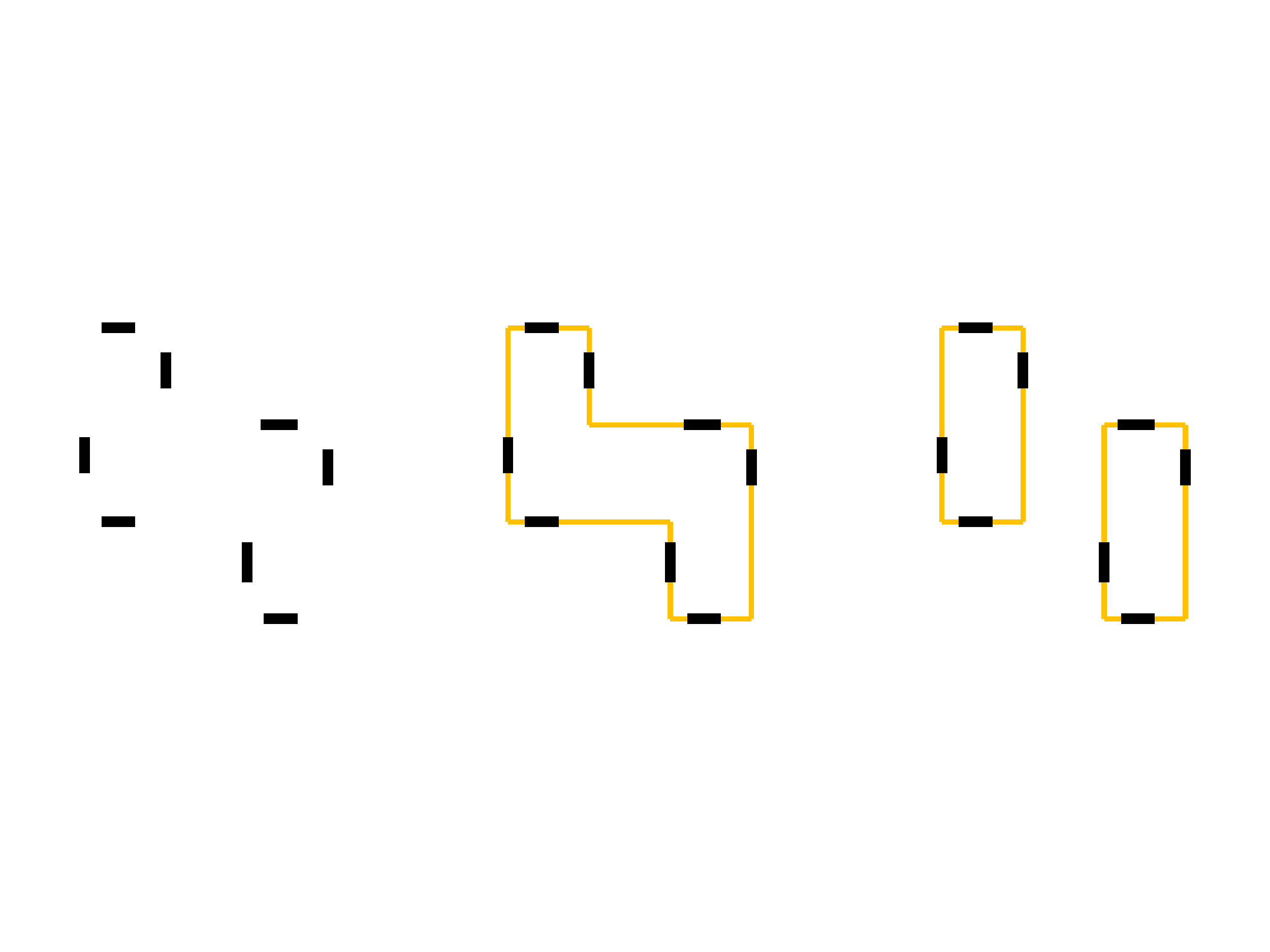}
\caption{\label{figure:simpleAmbi}A configuration of markers 
ambiguous under the ``single marker per face'' markup strategy.  
Two polyhedral interpretations of this 
marker set are shown.}
\end{figure}

\section{Reconstruction Results.}
\label{sec:results}

\begin{figure*}
\centering
\includegraphics[width=2.0in,height=1.8in]{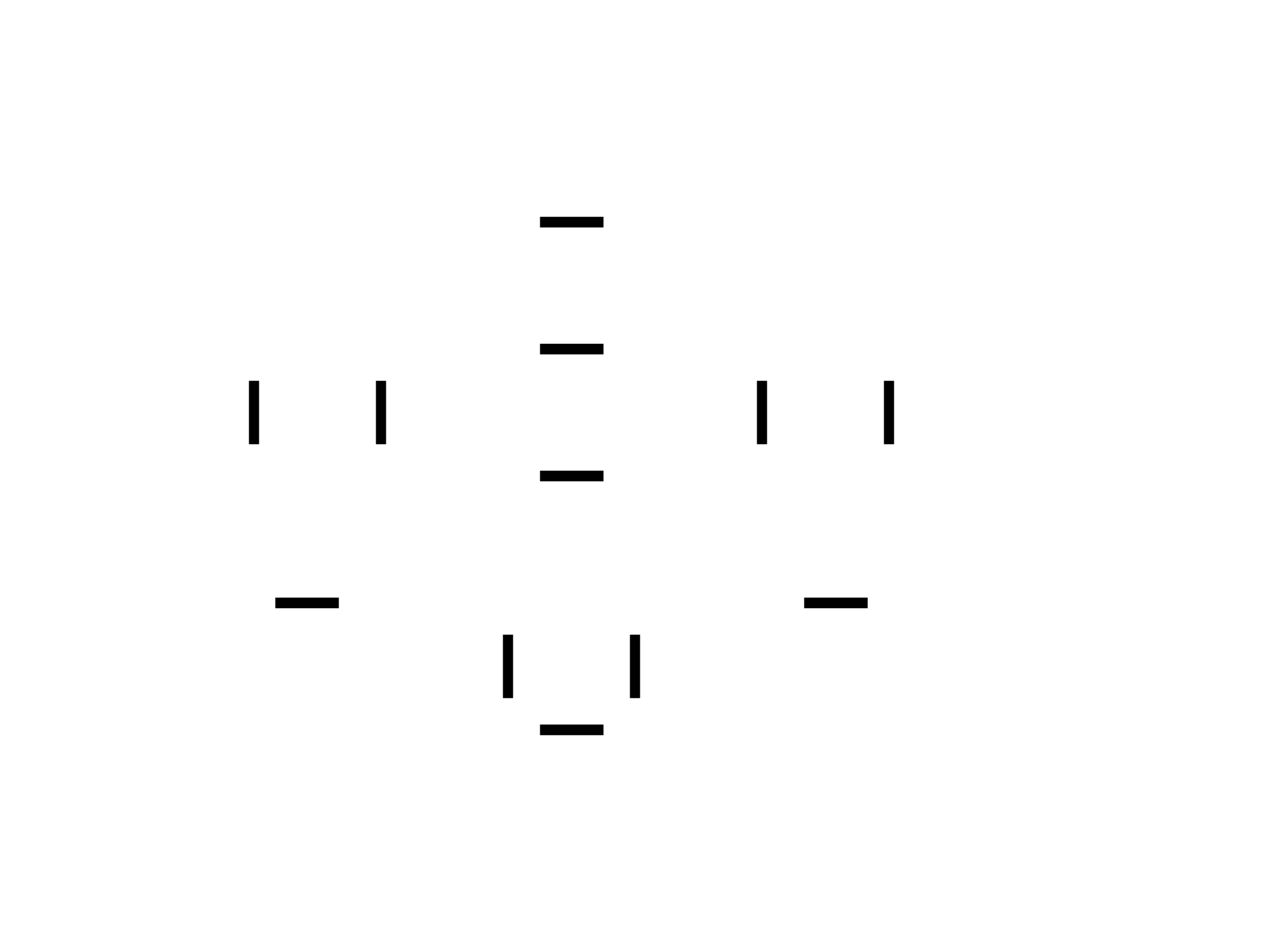}
\includegraphics[width=2.0in,height=1.8in]{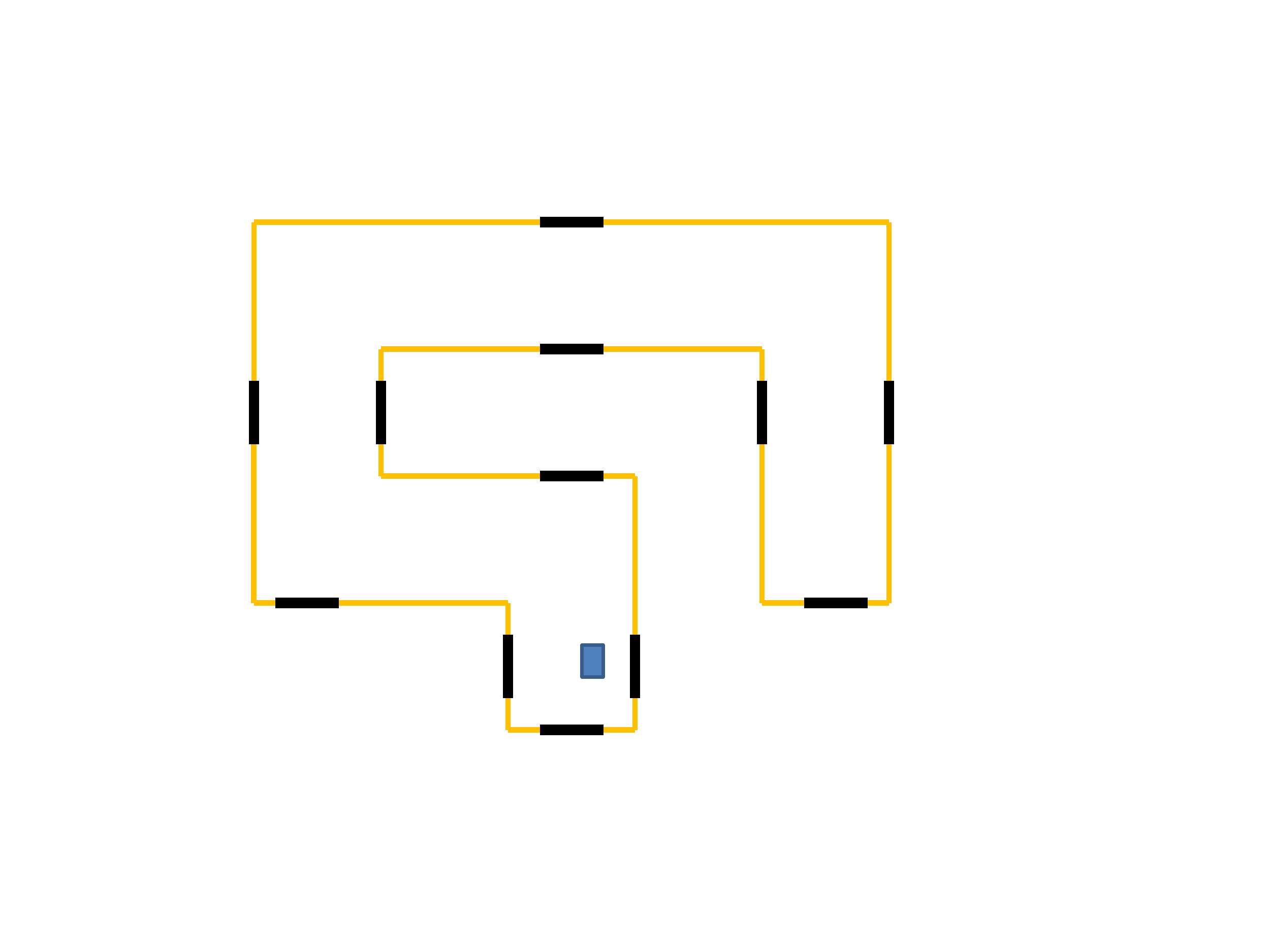}
\includegraphics[width=2.0in,height=1.8in]{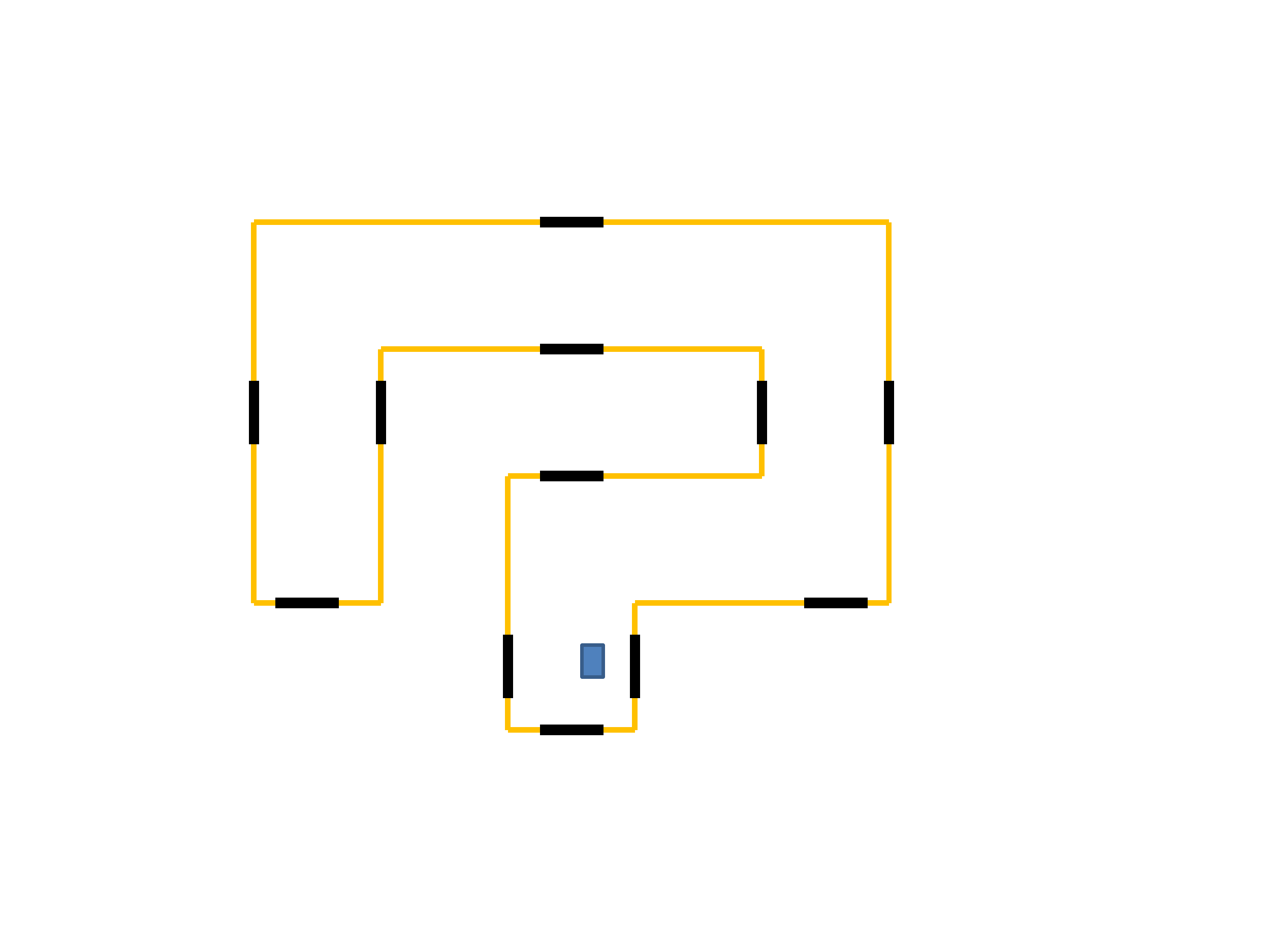}
\caption{\label{figure:ambiguousMarkers}A configuration of markers 
that supports two different, but topologically equivalent, interpretations
under the ``single marker per face'' markup system.} 
\end{figure*}

Our goal is to understand under what circumstances unambiguous 
reconstruction is possible from complete knowledge of all marker
poses and metadata. At the same time we would like the user's
task to be as simple as possible. We have two starting points, one easy
for the user but supporting reconstruction of only a small class of 
polyhedra, the other general but burdensome on the user.

\begin{theorem}[Convex polyhedron] 
\label{convex}
A marker-per-face strategy with point-plane metadata 
is sufficient to uniquely specify an arbitrary convex polyhedron.
\end{theorem} 

\begin{theorem}[Dense markup] 
\label{dense}
For any scene containing finitely many polyhedra, 
if markers are placed sufficiently densely, 
then the scene can be reconstructed. More precisely, for 
any scene containing a finite number of polyhedra, 
there is an $\epsilon$ such that if
every point is within $\epsilon$ of a marker, then the scene can be
unambiguously reconstructed.
\end{theorem} 

The problem with the marker-per-face approach is that 
Theorem \ref{nonConvex} and Theorem \ref{multConvex} show that 
both convexity, and the requirement that the scene contains only 
one convex polyhedron, are necessary to guarantee unambiguous 
reconstruction. The two main concerns with respect to
dense placement are that while there is a constructive means to
determine $\epsilon$, as the proof shows, it is not easy 
for a user to visually determine what density suffices, and it requires
users to place many more markers than are necessary, including in
places that may be hard to reach or even hidden.

The rest of this section pursues ways of enabling users to specify
broader classes of models while keeping
the burden on the user low in terms of 
the number of markers a user must place, the complexity of the instructions,
and the precision with which the user must place the markers. 

\begin{figure*}
\centering
\includegraphics[width=2.0in,height=1.8in,viewport=160 150 540 450,clip]{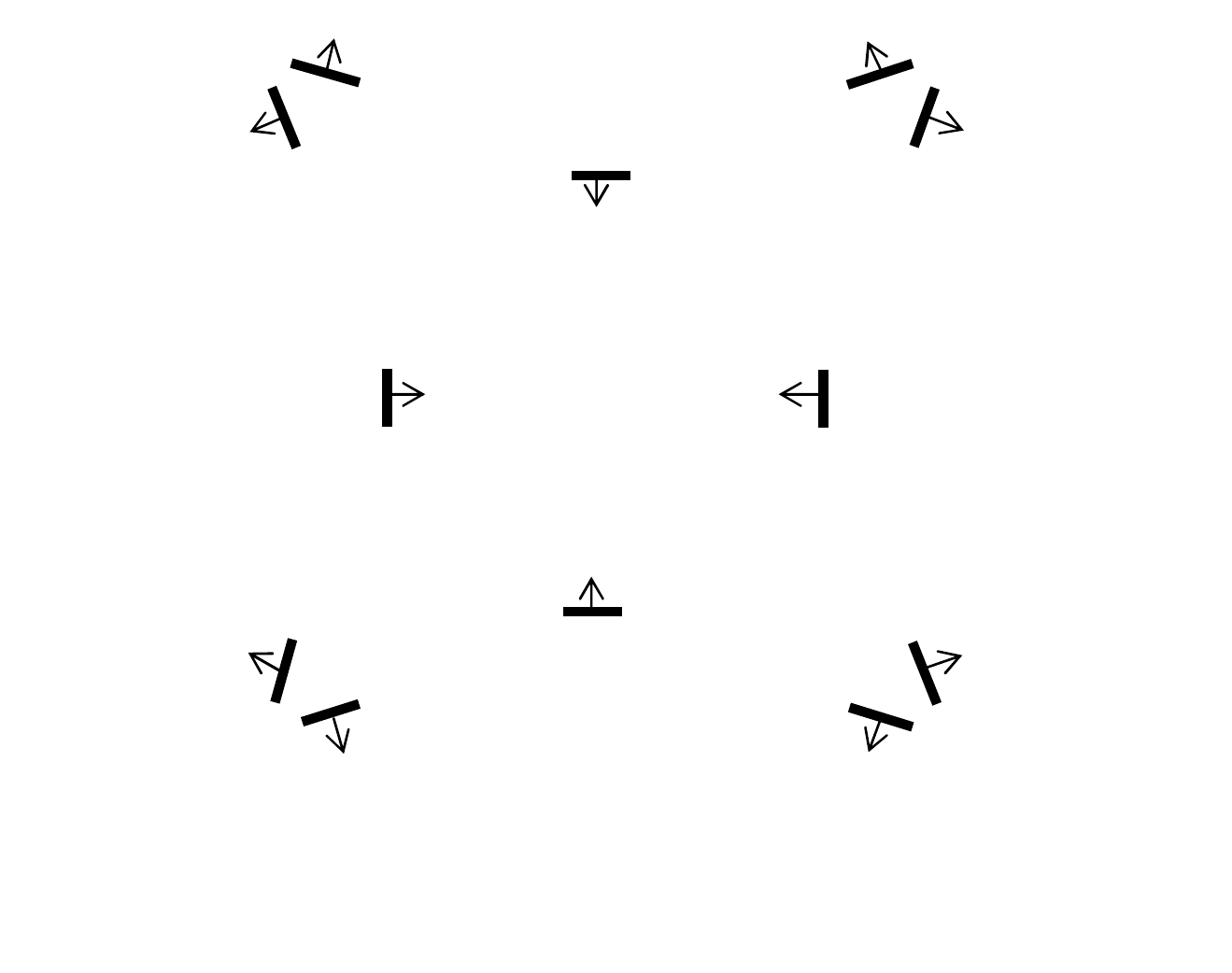}
\includegraphics[width=2.0in,height=1.8in,viewport=160 150 540 450,clip]{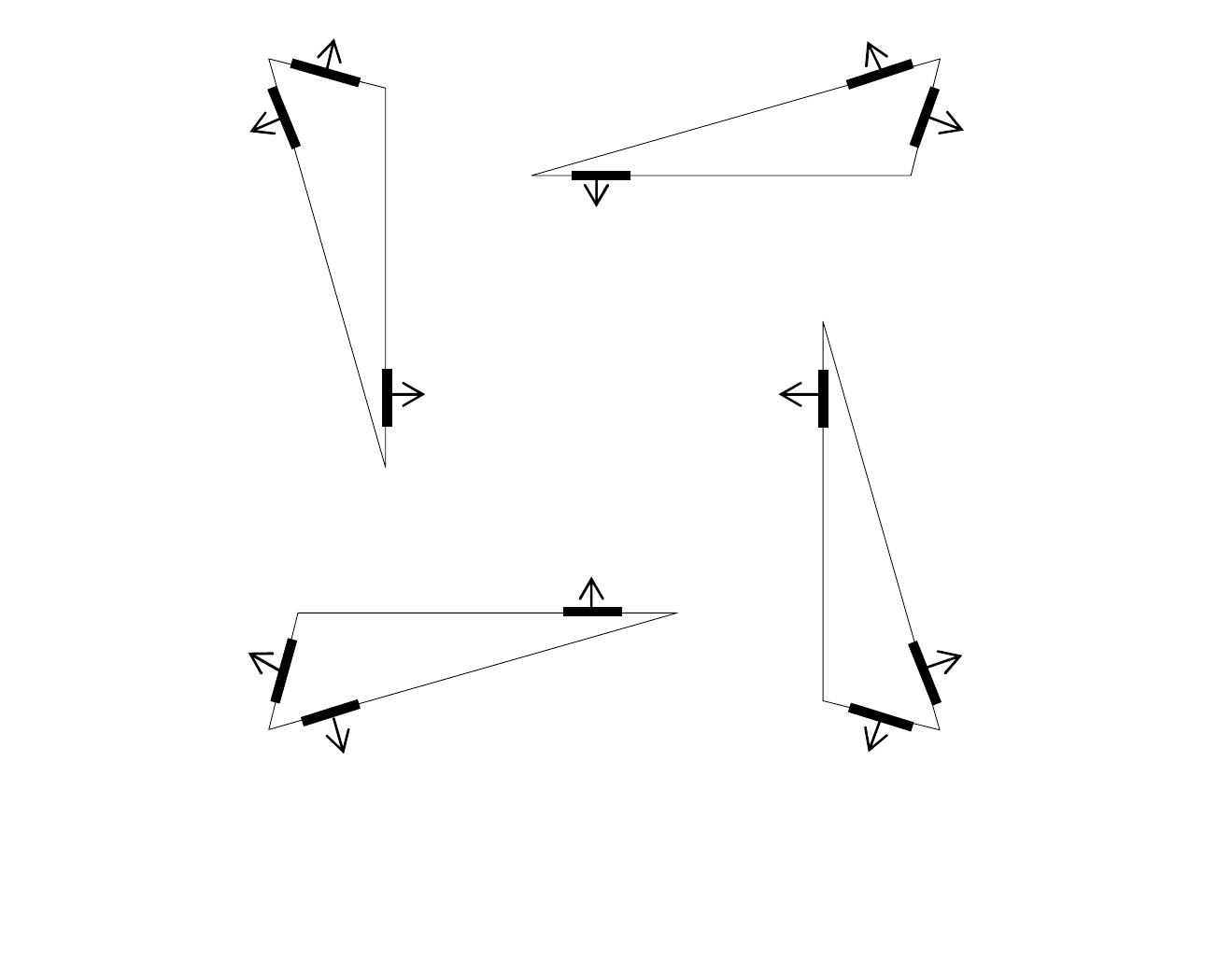}
\includegraphics[width=2.0in,height=1.8in,viewport=160 150 540 450,clip]{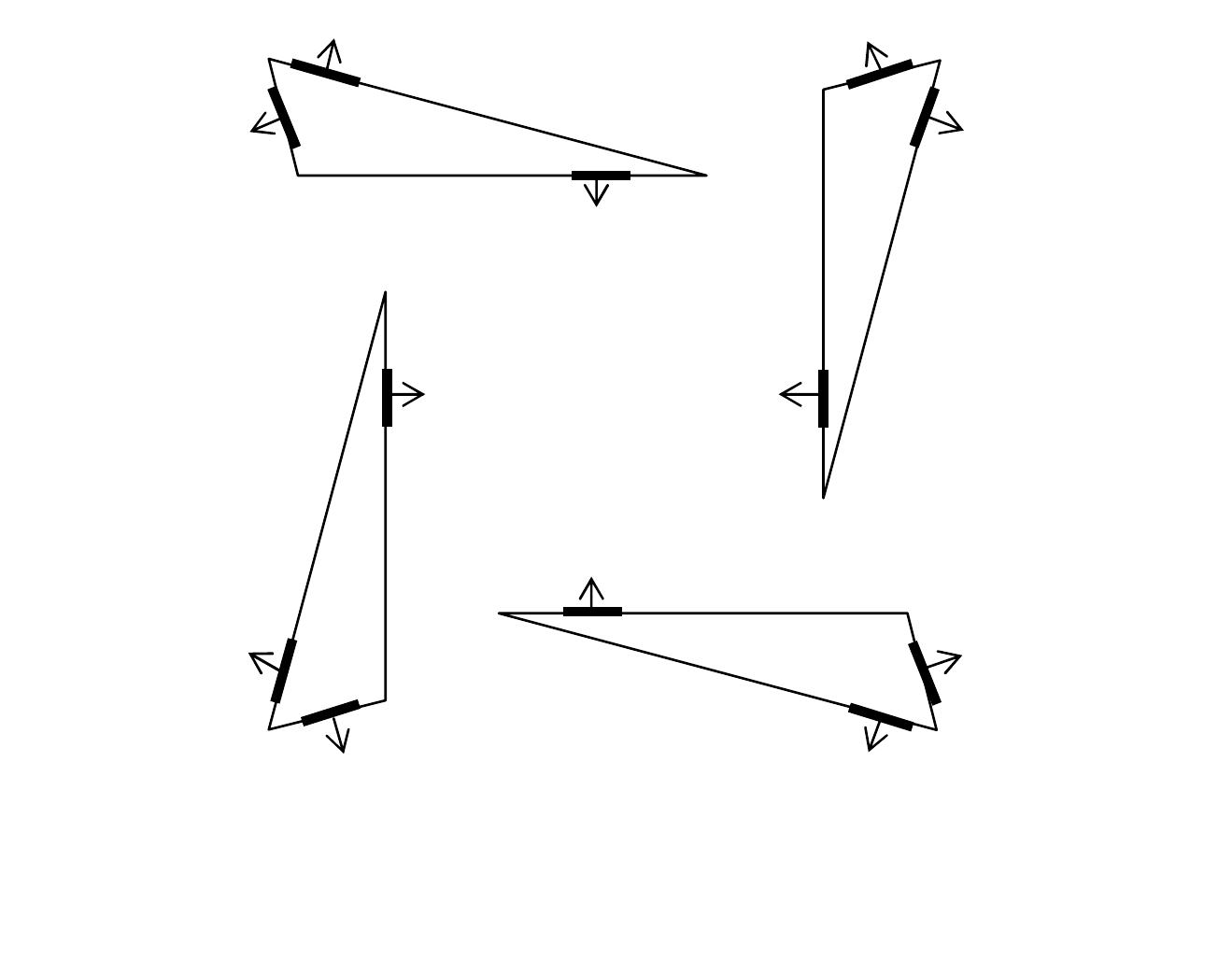}
\caption{\label{figure:ambiguousConvex}A configuration of markers that is 
ambiguous for the `collection of non-intersecting convex polygons' model 
class with the `single marker per edge' markup rule.}
\end{figure*}


Short proofs follow or preface the statement of the results. Longer proofs are
contained in Appendix \ref{proofs}. For example, the proofs of both
Theorem \ref{convex} and Theorem \ref{dense} are contained in Appendix
\ref{proofs}.

\subsection{Point-plane marker-per-face results.}
\label{subsec:pointPlane}

Theorem \ref{convex} states that 
a marker-per-face strategy with point-plane metadata suffices to 
uniquely reconstruct a convex polyhedron.
Convexity is required, however, to avoid ambiguity. 

\begin{theorem} 
\label{nonConvex}
A marker-per-face strategy with point-plane metadata 
is not sufficient to uniquely specify a non-convex polyhedron.
\end{theorem} 

\begin{proof} 
Figure \ref{figure:simpleAmbi} provides a simple counterexample. 
This example forms the basis for a more complex arrangement of
markers that are consistent with two different polyhedra,
shown in Figure \ref{figure:ambiguousMarkers}, that
have the same topology.
Both examples can be extended to three dimensions by placing two
markers on both an upper and lower face in the positions marked
by the gray squares.
\end{proof}

Similarly, that there is only one polyhedron in the scene is required
to avoid ambiguity.

\begin{theorem} 
\label{multConvex}
A marker-per-face strategy with point-plane metadata 
is not sufficient to uniquely specify multiple convex polyhedra.
\end{theorem} 

\begin{proof}
Figure \ref{figure:ambiguousConvex}.
\end{proof}

Theorem \ref{nonConvex} and Theorem \ref{multConvex} hold 
for point-normal metadata as well.

The ambiguities in the case of multiple convex polyhedra is easily
overcome by adding a polyhedron ID to the metadata so that all markers
placed on the same polyhedron have the same polyhedron ID and markers
placed on different polyhedra have different polyhedron IDs. The user's
task is now slightly more burdensome in that the user must keep track
of different sets of markers for different polyhedra in the scene, but
that is a relatively light additional burden.

\begin{theorem} 
\label{multConvexWithPolyID}
A marker-per-face strategy with point-plane and polyhedron ID metadata 
suffices to uniquely specify multiple convex polyhedra.
\end{theorem} 


We obtain more general results by considering elaborations of simple
polyhedra.

\subsection{Markup Containing Elaboration Information.}
\label{sec:elabResults}

We begin by stating the most general elaboration result we have and 
then discuss situations in which less elaborate markup suffices.

\begin{theorem}[Hierarchical elaboration] 
\label{genElab}
Any polyhedral scene that can be constructed by starting with
a finite set of convex polyhedra and iteratively elaborating faces with
orthogonal convex polygonal extrusions and intrusions can be unambiguously
reconstructed from 
a marker-per-face strategy with point-plane metadata together with
the following additional metadata that precisely specifies an 
elaboration hierarchy: 
each of the polyhedra from which the scene will be built must
have its own polyhedron ID, and each of its faces identified as
faces of the initial set of polyhedra from which the construction is made,
and each elaboration has its own ID
and the markers for an elaboration indicate on which base face the
elaboration is made.
\end{theorem} 

By including additional ordering information in the metadata, 
extrusions and intrusions from nonconvex polygons can also be handled.

\begin{theorem}
Theorem \ref{genElab} can be extended to the case of orthogonal nonconvex
polygonal intrusions and extrusions if, for each elaboration,
the markers on the faces perpendicular to the base of the elaboration
include additional metadata providing an ordering, say clockwise, of 
these faces.
\end{theorem}

A few remarks are in order. 
\begin{itemize}
\item Separate IDs are required for each elaboration in order
to handle cases in which the elaboration and markup resemble
that shown in Figure \ref{figure:ambiguousConvex}.
\item A given scene can often be constructed by different
series of elaborations. This nonuniqueness of the hierarchical 
elaboration supports different users' conceptualization of how a
scene can be built, but is also potentially a source of confusion
when a user considers how best to mark up a scene.
\item Because the different polyhedra each have their own ID, this
proof applies to scenes in which the different polyhedra intersect. 
The model can be rendered as intersecting polyhedra, but for efficiency
reasons it may be advisable to simplify the resulting model by
removing all regions internal to the final collection of polyhedra.
While this capability extends the class of scenes we can reconstruct,
it can add to the confusion of a user as to how best to markup a scene. 
One subcase of intersecting polyhedra is polyhedra that can be obtained
from simpler polyhedra by gluing along their boundaries.
\end{itemize}

Many real world objects and spaces are orthogonal elaborations of
convex polyhedra.
An even larger class of objects is well approximated by such polyhedra.
When the set of extrusions and intrusions in each face are derived
from a set of non-intersecting rectangles that are axis aligned for
some choice of two perpendicular axes, the separate elaborations do
not need different elaboration IDs.

\begin{theorem} 
\label{simpleExtrMulti}
Let $\cal P$ be a set of polyhedra obtained from a set of convex polyhedron
by intrusions and extrusions, where all of the intrusions 
and extrusions in each face are non-intersecting, orthogonal
rectangular intrusions. This set of polyhedra can be uniquely 
reconstructed from a marker-per-face strategy with point-normal metadata 
in which 
\begin{itemize}
\item all markers defining a single polyhedron have an associated
polyhedron ID, 
\item the markers specify the base faces of each polyhedron, and
\item all markers indicating an elaboration specify the base face in which 
the elaboration is done.
\end{itemize}
\end{theorem} 

Even the class of real world objects and spaces that can be obtained
simply by orthogonal intrusions of a box is reasonably large.
In this case, a simple 
marker-per-face strategy with point-normal metadata suffices
without the need for additional metadata.

\begin{theorem} 
\label{oneFaceIntrusionsPN}
Let $P$ be a polyhedron obtained from a box by a set of 
non-intersecting rectangular intrusions made in the interior of one face and
aligned with the sides of this face. The polyhedron $P$ can be
uniquely reconstructed from 
a marker-per-face strategy with point-normal metadata.
\end{theorem} 

Theorem \ref{oneFaceIntrusionsPN} can be strengthened to only
need point-plane metadata, though the proof is more involved.

\begin{theorem} 
\label{oneFaceIntrusions}
Let $P$ be a polyhedron obtained from a box by a set of 
non-intersecting rectangular intrusions made in the interior of one face and
aligned with the sides of this face. The polyhedron $P$ can be
uniquely reconstructed from 
a marker-per-face strategy with point-plane metadata.
\end{theorem} 

\subsection{Other sorts of markup.}
\label{sec:other}

As mentioned in Section \ref{sec:relWork}, researchers have looked at
polyhedral reconstruction from vertex data. We are less interested
in such reconstruction results because they require placement of
markers in hard to reach, or even hidden, places. Nevertheless we
state a few results.

O'Rourke gave an elegantly simple algorithm \cite{ORourke88} that reconstructs
orthogonal polygons from vertex data provided there are
no degenerate vertices connecting two collinear segments. This
condition together with orthogonality means that from every vertex
extends exactly one horizontal segment and exactly one vertical segment.
The reconstruction is then obtained using a straightforward 
{\it connect-the-dots} algorithm: Consider each row of vertices.
In each row, the first vertex must connect to the second, the third
to the fourth, and so on. Connecting the vertices of each column
in pairs in the analogous way, completes the argument. 
O'Rourke's algorithm can be extended to three-dimensions, again
with a prohibition against no $180^\circ$ vertices. This prohibition
is more stringent in the three-dimensional case in that it rules
out shapes we might like to cover such as one brick lying 
perpendicularly across another. 

\begin{theorem}
An orthogonal polyhedron in which all vertices do not connect any
two collinear segments can be uniquely and quickly reconstructed via
a connect-the-dots algorithm.
\end{theorem}

Reconstruction of a much more general class of polyhedra is possible
from vertex data together with some additional metadata.

\begin{theorem}
Any polyhedral scene in which every face is convex can be reconstructed
from a vertex markup with metadata that includes the face ID of all faces
meeting at that vertex. 
\end{theorem}

\begin{proof}
Each face can be reconstructed by taking the convex hull of all of
the vertices that are marked as associated with that face.
\end{proof}

General polyhedra can be reconstructed if the metadata also includes
ordering information.

\begin{theorem}
Any polyhedral scene can be reconstructed from vertex markup if
the metadata includes not only the face ID but also ordering information
that specifies the clockwise order in which the vertices would be
traversed when walking around the boundary of the face.
\end{theorem}

\begin{proof}
The boundary of each face can be reconstructed by placing a segment
between each vertex and the next one in the ordering. 
\end{proof}

%

\section{Future Work.} 
\label{sec:future}

We have explored a number of markup strategies with various
types of metadata in search of markup systems that are 
are easy for users to carry out and that unambiguously 
reconstruct an attractive class of models.
We are in the process of improving our system on the basis of these
results. We look forward to working with artists to test our system, 
and to learn which strategies are most intuitive and effective for such users.

A number of open research questions remain. We did not
consider robustness issues. When the marker pose is obtained from
camera estimates, there are interesting questions as to the best markup
strategies in the face of inaccurate estimates or failing to detect
one or more markers altogether. Robustness in light of user errors
in entering metadata is also interesting. A number
of our results are not sharp, such as Theorem \ref{genElab}, in that
polyhedra not covered by the theorem can be reconstructed
from the marker strategy and metadata. The question is whether
the statement can be extended to a well-defined class of polyhedra.
Similarly, in some cases, polyhedral scenes could be reconstructed
with less metadata. 
With care, some classes of polyhedra obtained by 
extensions from polygons with boundaries intersecting the boundary
of the face could be covered. 
Consideration of non-orthogonal extrusions may be fruitful. 
We reconstructed a few non-polyhedral shapes such as cylinders.
Extending these results to more general parametrized shapes would yield 
a much richer class of scenes that can be reconstructed.
We primarily explored a marker-per-face strategy leaving open 
the question of whether a few markers placed more strategically would
enable unambiguous reconstruction of a more general class of polyhedra.
In general, we continue to 
search for practical markup strategies that yield unambiguous
reconstruction of scenes from a sparse set of markers. 

\bibliographystyle{abbrv}
\bibliography{copyOfPantheia}





\appendix

\section{Proofs for results in Section \ref{sec:results}.}
\label{proofs}

This appendix contains proofs of results in Section \ref{sec:results}.
It begins with proofs of Theorem \ref{convex} and Theorem \ref{dense}.

To prove Theorem \ref{convex}, we make use of the following Lemma
and Corollary.

\begin{lemma}
\label{oneSide}
A convex polyhedron $P$ lies entirely on one side of the plane defined by
any of its faces.
\end{lemma}

\begin{proof}
By contradiction. 
Suppose there are points of the polyhedron $P$ on both sides of the 
plane $R$ defined by a face $F$. We show that
such a polyhedron cannot be convex.
Let $p$ be a point on face $F$. The face bounds the polyhedron on
one side, so there exists an $\epsilon > 0$ such that there are no
points within $\epsilon$ of $p$ in one of the open half-spaces $H$ 
defined by $R$.  
Since by hypothesis the polyhedron contains points on both sides
of the plane $R$, let $p_1$ be a point in the polyhedron contained in
the half-space $H$. Since the polyhedron is convex by hypothesis, 
the entire line segment between $p$ and $p_1$ must be contained
in the polyhedron. This line segment contains points within $\epsilon$
of $p$. We have reached a contradiction, thus our supposition that
there are points of $P$ on both sides of $F$ must be false.
\end{proof}

\begin{theorem}[Corollary to Lemma \ref{oneSide}]
\label{faceIntersections}
In a convex polyhedron, 
the intersection of any plane $R$ defined by a face $F$ of a convex polyhedron
with any other face $F'$ of the polyhedron must be contained entirely within an
edge of $F'$.
\end{theorem}

\begin{proof}
If the plane $R$ intersected the face $F'$ anywhere other than in
an edge of $F'$, the face $F'$ would contain points on both
sides of the plane $R$, but that is impossible by Lemma \ref{oneSide}.
\end{proof}

\begin{proof}[Proof of Theorem \ref{convex}, Convex Polyhedron]
For each marker, determine the half-space which contains all 
the other markers, as per Lemma \ref{oneSide}.
Take the intersection of all these half-spaces. 
The intersection of convex sets is convex, and half-spaces
are convex, so this intersection is convex. 

We want to prove that this convex polyhedron $P_I$ is the same
as the original one, $P_0$, that was marked. To do so, it suffices to 
show that no other convex polyhedron has faces with the same set 
of bounding planes. By Lemma \ref{oneSide}, any convex polyhedron
lies entirely on one side of the plane defined by any of its 
faces. By this Lemma, any other convex polyhedron defined
by these faces must be contained in $P_I$. To prove the converse,
suppose $P_0$ were strictly contained in $P_I$. In this case its face set 
must be different. Because the face planes of $P_I$ and $P_0$
are identical, and $P_0\subset P_I$, a face $F_0$ of $P_0$ must be
strictly contained a face $F_I$ of $P_I$. Let $E_0$ be an edge of
$F_0$ that is not an edge of $F_I$; in particular, $E_0$ is a line
segment in the interior of $F_I$. The edge $E_0$
must arise as the intersection of two faces of $P_0$, so must 
be contained in the intersection of two face planes $R_I$ and $R_J$
of $P_0$, one for the face $F_I$ and one for another face, $F_J$, 
of $P_0$. The two planes $R_I$ and $R_J$ are also face planes for
$P_I$. By construction $R_J$ intersects $F_I$ in its interior, but
that contradicts Corollary \ref{faceIntersections}. Thus $P_0$ 
and $P_I$ must be the same polyhedron.
\end{proof}

We now turn to proving Theorem \ref{dense}.
\begin{proof}[Proof of Theorem \ref{dense}, Dense Markup.]
The markers define a finite number of planes. Partition each plane
into regions along lines of intersection with all other planes.
Remove from consideration all infinite regions. Each of the remaining
regions contains a disk of positive radius for some radius $\epsilon$.
Take the minimum $\epsilon_{min}$ of these positive radii over all of 
the regions. Since the set of regions is a finite set, this minimum 
exists and is positive. Unambiguous reconstruction of the scene can
be done from any marker strategy that places at least one
marker within distance $d < \epsilon_{min}$ of every point on every
face in the polyhedral scene. The reconstruction algorithm simply keeps
every region in which there is a marker and discards the rest.
\end{proof}

\subsubsection{Proofs for Elaborations on a Convex Polyhedron.}

This section contains proofs for the results stated in Section
\ref{sec:elabResults}.

\begin{proof}[Proof of Theorem \ref{genElab}.]
The marker metadata specified in Theorem \ref{genElab} specified the
entire hierarchy of elaborations by specifying to which face each
elaboration is made.
The marker metadata tells us which markers lie on faces belonging to 
the set of initial polyhedra on which the hierarchical construction is made. 
Start by identifying these markers. Because each initial polyhedron has
its own ID specified by the marker metadata, we can
reconstruct the initial set of polyhedra by Theorem \ref{convex}.

For each face on this initial set of polyhedra, identify any elaborations
made to this face as indicated by the marker metadata. Each elaboration
has its own ID, so to construct the elaboration, identify all markers
with that ID that are perpendicular to the face. Projecting the markers
onto the face defines a polygon. If there are any markers with this
elaboration ID parallel to the face, they determine the depth of 
the extrusion or intrusion. If no such marker exists, the elaboration
must be an intrusion and must go all the way through the polyhedron.
Repeat this process for all faces in the initial set of polyhedra.
The order in which each face is considered is irrelevant. This
process is then repeated for the next level of the hierarchy: each 
face just created is considered in turn and any elaborations to
those faces are determined. This process is continues for all levels
in the elaboration hierarchy until there are all elaborations have
been incorporated. 
\end{proof}


We now prove some two-dimensional results to support the proofs of 
Theorem \ref{oneFaceIntrusionsPN} and Theorem  \ref{oneFaceIntrusions}.
We first prove a uniqueness results and then
give an algorithm that performs the reconstruction.
The proofs and algorithms make use of the following concepts.


\begin{Definition} {\rm A set of polygons is {\em consistent} with a set of 
markers if every marker is on (and aligned with) an edge of one
of the polygons.} 
\end{Definition}

\begin{Definition} {\rm A consistent set of polygons is {\em fully consistent} 
with a set of markers if in addition every edge in the set of polygons 
has at least one marker on it (and aligned with it).}
\end{Definition}

\begin{Definition} {\rm A marker is considered {\em in the interior} of a 
polygon if its center is in the interior of the polygon, or it 
is on an edge of the polygon but not aligned with it. Markers on 
an edge, and aligned with it, are not considered in the interior.}
\end{Definition}


Consider the case of orthogonal non-intersecting
convex polygons; in other
words, we are considering a set of axis-aligned rectangles. 
Suppose this collection of rectangles is marked up using a 
marker-per-face strategy with point normal data.
The markers can be divided into four classes depending on the 
direction of their outward facing normal, 
$\cal L$, $\cal R$, $\cal T$, and $\cal B$, 
where the normals for the markers in $\cal L$ point to the left, 
to the right in $\cal R$, to the top in $\cal T$, to the bottom in $\cal B$. 

The rest of this section shows that a set of non-intersecting orthogonal
rectangles can be uniquely reconstructed from any marker-per-face strategy
with point-normal data. The idea behind the proof is that, in such a markup of
an arbitrary collection of axis aligned rectangles,
if we can find one rectangle $P$ that belongs to a set 
of rectangles that is fully consistent with the marker set, and if we
can show that every fully consistent set of rectangles contains this rectangle,
by induction there is only one set of rectangles fully consistent with the 
markers. We begin with a lemma that finds such a rectangle $P$.

\begin{lemma}
\label{uniqueRect}
Let $\cal M$ be an arbitrary set of markers placed according to a 
marker-per-face strategy with associated point-normal data on an arbitrary
set $S$ of non-intersecting, orthogonal rectangles. There is a unique 
rectangle $P$ consistent with the left-most (and top-most among these 
if there is a tie) marker $M_L\in\cal L$ 
that is contained in all sets of rectangles that are fully consistent 
with this set of markers. Moreover, the rectangle $P$ is the only 
rectangle consistent with marker $M_L$ that has markers on and 
aligned with each of its edges and no markers in its interior.
\end{lemma}

\begin{proof}
Suppose $M_L$ lies on a rectangle $P$ that is a member of a fully 
consistent set of non-intersecting rectangles $S$ with respect to 
$\cal M$, and that $M_L$ also lies on a rectangle $P'$ that is a 
member of another fully consistent set of non-intersecting rectangles $S'$. 
Let $M_R$ be the topmost marker in $\cal R$ on the 
right side of of $P$, and $M_T$ and $M_B$ be the left-most markers 
in $\cal T$ and $\cal B$ on the top and bottom, respectively, of $P$. 
(See Figure \ref{figure::orthoProof}.) Similarly, let $M_R'$ be 
the topmost marker on the right side of of $P'$, and $M_T'$ and $M_B'$ 
be the left-most markers on the top and bottom, respectively, of $P'$. 
Given the way we chose $M_L$, it must be the top-most marker on the left 
side of both $P$ and $P'$. Furthermore, without loss of generality, 
we may assume that marker $M_T'$ is either to the right of $M_T$, 
or $M_T'$ and $M_T$ have the same horizontal component. (Otherwise,
switch $P$ and $P'$.) The marker $M_T'$ must not be higher than $M_T$
since otherwise $M_T$ would be in $P'$, which is impossible
since $P'$ is part of a consistent set of non-intersecting rectangles. 
Unless $M_T$ and $M_T'$ are 
the same marker, they cannot be equal in height, because then $M_T$ 
would be a top marker on $P'$ to the left of $M_T'$. We are left 
with two cases to consider: either $M_T'$ is to the right and 
lower than $M_T$, or $M_T = M_T'$. 

Consider the case in which $M_T'$ is to the right 
and lower than $M_T$. The marker $M_T'$ must be to the right of 
$M_R$, since $M_T'$ cannot be in $P$. And, by definition,
$M_T'$ is above $M_L$. There are two cases: 
either $M_R$ is above $M_L$ or $M_R$ is below $M_L$, as illustrated
in Figure \ref{figure::orthoProof}. 
(If $M_L$ and $M_R$ were equal in height, the $P'$ would
contain $M_R$.) Suppose $M_R$ is above $M_L$. 
Then $P'$ must be below $M_R$. But now $M_R$ cannot be part of 
any rectangle in a set consistent with $P'$; there are no markers 
that could form the bottom of such a rectangle because there are no 
markers in $\cal B$ below and to the left of $M_R$ and above $P'$ 
since $P$ contains no markers in its interior.
Similarly, suppose $M_R$ is below $M_L$. Then $P'$ must be above $M_R$. 
But then $M_R$ cannot be part of a rectangle in a set consistent with
$P'$ because there are no markers in $\cal T$ below $P'$ and 
above and to the left of $M_R$. So the only alternative left open
to us is that $M_T'$ and $M_T$ are the same marker. Similar arguments
show that $M_B$ and $M_B'$ must be the same
marker, and that $M_R$ and $M_R'$ must be the same marker.
Thus $P$ is unique.
\end{proof}

\begin{figure}
\centering
\includegraphics[width=4.0in,height=3.0in]{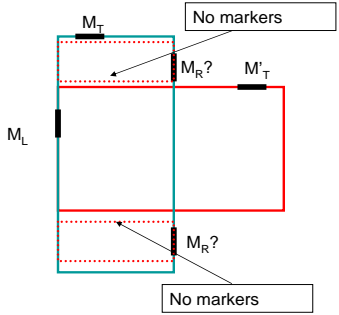}
\caption{\label{figure::orthoProof} Diagram supporting the proof
of Lemma \ref{uniqueRect}. }
\end{figure}

\begin{theorem}
\label{orthConvex2D}
A set $\cal R$ of non-intersecting axis-aligned rectangles
can be uniquely reconstructed from any set of markers $\cal M$ 
placed according to a marker-per-face strategy with point-normal data.
\end{theorem}

\begin{proof}
We proceed by induction 
on the number of rectangles (which we do not assume is known). 

Let $O$ be the bounding orthogonal rectangle, the smallest orthogonal
rectangle that contains all the markers. Each edge of $O$ has at least
one marker on it, and aligned with it, because by definition of $O$
each edge of $O$ must intersect the set of rectangles, and orthogonality
guarantees that this intersection must be in an edge, not just a point; 
we can find $O$ from the markers by finding all markers that define a 
line which is a boundary of the set -  
a line is on the boundary if all markers
are in only one of the half-planes defined by the line - 
and letting $O$ be the intersection of all of these half planes. 

{\bf Case $n = 1$:} If there is only one rectangle, 
$O$ will be that rectangle. We know we are in this case if all
of the markers are on and aligned with an edge of $O$.

{\bf Case $n = k$:} By Lemma \ref{uniqueRect}, any fully
consistent set of rectangles must contain the rectangle $P$. 
Now exclude from $\cal M$ any markers that are consistent with $P$. 
The remaining markers must be consistent with a set of $k-1$ rectangles. By
induction, there is only one such set. 

Thus there is only one set of rectangles fully consistent with the set 
of markers $\cal M$.
\end{proof}

The orthogonality assumption in Theorem \ref{orthConvex2D} is necessary.

\begin{theorem}
\label{multConvex2D}
A set of non-intersecting convex polygons cannot always
be reconstructed from a set of markers that includes at least one marker 
per edge.
\end{theorem}

\begin{proof}
A counterexample is shown in Figure~\ref{figure:ambiguousConvex}.
\end{proof}

\subsubsection{A reconstruction algorithm for Theorem \ref{orthConvex2D}.}

This section describes a constructive algorithm for finding the unique 
set of rectangles fully consistent with a marker set. The algorithm 
proceeds by finding a rectangle $P$ that is consistent with the top 
left-most marker $M_L$. The rest of this subsection describes how that 
is done.  Once it is done, all markers consistent with $P$ are ignored, 
the new top left-most marker is considered, and a new rectangle is 
found. The algorithms continues until all markers have been 
taken care of. 

We set out to find a rectangle consistent with $M_L$ that 
has markers on and aligned with each edge and no markers in 
its interior.
Consider all markers in $\cal T$ that are to the right and 
above $M_L$; one of these must be able to form a rectangle with $M_L$ 
that is part of a fully consistent set of rectangles for the 
marker set. Consider first the left-most of these markers and, 
if there are multiple left-most markers, take the top-most one of
these. Call this marker $M$. 

Partition $O$ into an irregular grid with lines defined by 
the markers in $\cal M$. Consider the smallest rectangle consistent 
with markers $M_L$ and $M$ that is made up of grid segments. 
If this rectangle contains any markers in its interior, $M$ 
cannot form a rectangle with $M_L$ that is part of a fully 
consistent set, so go on to consider the next left-most marker. 
Eventually a marker $M_{Tcand}$ is found such that the smallest grid-aligned 
rectangle $R$ containing $M_L$ and $M_{Tcand}$ does not have any markers in its 
interior. (There must be at least one, since by assumption, the markers
mark the original set of non-intersecting rectangles, the set we are 
trying to find.)  


Check if there are markers on, and aligned with all sides of $R$. 
If there are, we take this rectangle to be $P$. Otherwise, expand 
the rectangle until this property holds as follows. Let $S$ be the set of 
markers in $\cal B$ that are below and to the right of $M_L$. The set $S$ 
cannot be empty since that would mean $M_L$ is not on a rectangle
that is part of a fully consistent set.
Among the markers in $S$, consider 
the subset $S_M$ of markers that do not have any markers in $S$ directly 
above them or above them to their left. The left-most marker in $S$ (and 
top-most, if there is more than one left-most) satisfies this criterion. 
Consider first the left-most (and top-most if there is a tie) of these markers,
$M_{Bcand}$. If the smallest rectangle consistent with $M_{Bcand}$, $M_{Tcand}$,
and $M_L$ has markers in its interior, go on to the next left-most
marker in $S_M$. If there are no markers in the interior, find the
left-most marker or markers in $\cal R$ and in the strip bounded 
on the top by $M_{Tcand}$, on the bottom by $M_{Bcand}$, and to the 
right of $M_L$. If there is no such marker, then 
$M_{Tcand}$, $M_L$, and $M_{Bcand}$, cannot form a marker consistent 
rectangle, so go on to consider the next candidate for
$M_{Bcand}$. If none of the possible candidates for 
$M_{Bcand}$ work, then try the next possible $M_{Tcand}$.
If there is no marker in the interior of the rectangle defined by
$M_{Tcand}$, $M_L$, and $M_{Bcand}$, find the left-most marker in $\cal R$
that is in the strip bounded 
on the top by $M_{Tcand}$, on the bottom by $M_{Bcand}$, and on the 
left by $M_L$. If there is no such marker, then 
$M_{Tcand}$, $M_L$, and $M_{Bcand}$, cannot form a marker consistent 
rectangle, so go on to consider the next candidate for
$M_{Bcand}$. Again, if none of the candidates $M_{Bcand}$ work, 
then consider the next candidate $M_{Tcand}$.
Eventually some set works, since we know the markers form
a fully consistent set for some set of rectangles.

\begin{Definition}
{\rm An infinite polygonal {\em cylinder} consists of a 3D region obtained
by infinitely extending in both directions the edges of a polygon 
(that lies in 3D space) perpendicular to the face of the polygon.}
\end{Definition}

\begin{Definition}
{\rm A polygonal {\em half-cylinder} consists of a 3D region obtained by 
infinitely extending in the edges of a polygon in one of the 
directions perpendicular to the face of the polygon.}
\end{Definition}

\begin{proof}[Proof of Corollary \ref{simpleExtrMulti}]
For each set of markers with the same polyhedron ID, construct
the faces of the base polyhedron from the markers with metadata
indicating that they are on one of the faces of the
base polyhedron. For each base face $F$, consider all markers $S_F$ 
that indicate that they are
on faces made by elaboration of this face. 
Find all such markers that are perpendicular to this face, 
and project them onto the plane defined by this face. By
hypothesis, the result defines a set of rectangles on the plane. Apply 
Theorem \ref{orthConvex2D} to obtain the uniquely
determined set of rectangles. Each rectangle corresponds to an 
intrusion or extrusion. Consider the infinite rectangular cylinder 
defined by each of the rectangles.
If there is no marker contained in $S_F$ in the interior of this
cylinder, this rectangle defines an
intrusion that goes all the way through the polyhedron. 
If there are markers in $S_F$ inside that cylinder, they must all
lie on the same plane, since rectangles defining the extrusions and 
intrusions do not intersect by hypothesis. This plane determines the 
extent of the extrusion or intrusion. The same construction can now
be applied to all elaborations of faces that resulted from this
level of elaboration, and can continue to be applied at each level
until all elaborations have been handled.
\end{proof}

\begin{proof}[Proof of Theorem \ref{oneFaceIntrusionsPN}]
Consider all $xy$-aligned markers. All of these markers lie on
faces that are either boundaries of the original box or
are faces obtained from making interior orthogonal rectangular
intrusions in the same face. Because the intrusions are in the
interior, all of the markers obtained by intrusion have $z$ 
coordinates strictly between those of the boundary regions. So
the $xy$ boundary planes can be obtained by taking the
extremal $z$ values from among the $xy$-aligned markers. The
same process can be repeated for $xz$- and $yz$-aligned 
markers to obtain the original box. 

Now consider all markers in the interior of the box. 
The markers fall into six classes determined by the
direction of their outward facing normal, 
$\cal L$, $\cal R$, $\cal T$, $\cal B$, $\cal F$, $\cal G$, 
where the normals for the markers in $\cal L$ point to the left, 
to the right in $\cal R$, to the top in $\cal T$, to the bottom in $\cal B$, 
to the front in $\cal F$, to the back in $\cal G$.
Since all interior markers come from intrusions, and these 
intrusions were all made in the same face, exactly one of
of these classes is empty. The boundary face with that class's
normal is the face in which the intrusions were made. 

To finish the construction, we project all of the non-boundary 
markers that are perpendicular to the face onto the face. The
result defines a set of non-intersecting axis-aligned rectangles.
To this set, apply Theorem \ref{orthConvex2D} to obtain the unique
fully consistent set of rectangles. Now, just as in the proof of
Theorem \ref{simpleExtrMulti}, for each rectangle, consider the
infinite rectangular cylinder it defines. If there are no non-boundary
markers inside that cylinder, then the intrusion goes all the way
through the box. If there are non-boundary markers inside the cylinder,
they must all be in the same plane since the set of rectangles does
not intersect. These markers define the depth of the intrusion.
\end{proof}

\begin{proof}[Proof of Theorem \ref{oneFaceIntrusions}]
The proof finds the faces of the bounding box in the same way as the proof 
for Theorem \ref{oneFaceIntrusionsPN}. Because we no longer have access to
the normal direction, the method to find the face in which the
intrusion was made is more involved. 

By hypothesis all of the interior markers
must be obtained by intrusion from one
face. Consider each face in turn. Among the interior markers parallel 
to a face, find the marker closest to the face. If there are no markers
parallel to a face, all intrusions must be from this face and must
go all the way through the box. If there is a marker, check to see if
there are perpendicular markers closer to the face. If not, this is 
not the face in which the intrusions were made. If so, this is the 
face, because if the closest marker were on the side, rather than the
bottom, of an intrusion, 
there would be no markers closer to the face than it. 

In this way we have determined the face
in which all the intrusions have been made. 
The proof finishes in exactly the same way as the proof for
Theorem \ref{oneFaceIntrusionsPN}. 
\end{proof}

\end{document}